\documentclass[12pt]{article}
\pdfoutput=1 
\usepackage[defblank]{paralist}
\usepackage[tbtags]{amsmath}
\usepackage{nicefrac}
\usepackage{titlesec}
\usepackage{multicol}
\usepackage{physics}
\usepackage[dvips]{graphicx}
\usepackage{eufrak}
\usepackage{mathrsfs}
\usepackage{url}
\usepackage[a4paper]{geometry}
\usepackage{color}
\usepackage{psfrag}
\usepackage{array}
\usepackage{booktabs}
\usepackage{epic}
\usepackage{eepic}
\usepackage{subfigure}
\usepackage{verbatim}
\usepackage{colortbl}
\usepackage{tikz}
\usepackage{pgf}
\DeclareGraphicsExtensions{.jpg,.pdf,.png}

\usepackage{fancyhdr}
\usepackage{mathtools}
\usepackage{braket}
\usepackage{MnSymbol}
\usepackage{amsthm}
\usepackage{hyperref}
\usepackage{afterpage}
\usepackage[english]{babel}
\usepackage{csquotes}
\usepackage{a4wide}
\usepackage{cancel}
\usepackage{lineno,hyperref}
\usepackage[inline]{enumitem}

\newtheorem{assumption}{Assumption}[section]
\newtheorem{axiom}{Axiom}[section]
\newtheorem{definition}{Definition}[section]

\newtheorem{proposition}{Proposition}[section]
\newtheorem{remark}{Remark}[section]
\newtheorem{theorem}{Theorem}[section]

\newcommand{\beq}{\begin{equation}}
\newcommand{\eeq}{\end{equation}}

\newcommand{\bdm}{\begin{displaymath}}
\newcommand{\edm}{\end{displaymath}}
\newcommand{\bea}{\begin{eqnarray}}
\newcommand{\eea}{\end{eqnarray}}
\numberwithin{equation}{section}

\begin{document}
\title{On the Reality of the Wavefunction}
\author{Omid Charrakh\thanks{Munich Center for Mathematical Philosophy, LMU Munich, Geschwister-Scholl-Platz 1, 80539 Munich (Germany) -- omid.charrakh@campus.lmu.de.}}
\date{}
\maketitle
\begin{abstract}
Within the Ontological Models Framework (OMF), Pusey, Barrett, and Rudolph (PBR) have given an argument by which they claimed that the epistemic view on the wavefunction should be ruled out. This study highlights an incorrect conclusion in PBR's arguments, which was made due to inadequacies in the definitions of OMF. To be precise, OMF models the ontology of the preparation procedure, but it does not model the ontology of the measurement device. Such an asymmetric treatment becomes problematic, in scenarios in which measurement devices have a quantum nature. Modifying the OMF's definition such that the ontology of the measurement device becomes included, we will see how PBR's result disappears. 
\end{abstract}

\tableofcontents

\section{Introduction}

In the foundations of quantum mechanics, the status of the wavefunction is one of the most debated issues. Associating a wavefunction with each system, QM provides very accurate predictions about the possible measurement outcomes. However, the axioms of QM are silent about the nature of the wavefunction, i.e., these axioms are logically consistent with several different ontological statuses of the wavefunction.

The Ontological Models Framework (OMF) introduced by Harrigan and Spekkens \cite{harrigan2010einstein}, is an attempt to provide a general formalism by which one can evaluate the ontology of the quantum theory and specifically the ontology of the wavefunction in realistic interpretations of QM.             

Leifer explains \cite{leifer2011can} that one can choose one of the following three philosophical stands toward the wavefunction's interpretation. 

\begin{enumerate}
\item Antirealist Epistemic: Wavefunctions are epistemic, and there is not any deeper reality behind them. 
\item Realist Epistemic: Wavefunctions are epistemic, but there is an underlying reality.
\item Realist Ontic: There is an underlying reality, and the wavefunctions correspond to the states of reality.
\end{enumerate}

Options one and two are called $\psi$-epistemic views. In these views, one interprets the wavefunction as a mathematical object which merely represents the agent's knowledge of the quantum systems. Option three is referred to as the $\psi$-ontic view, within which the wavefunction is considered to be a \textit{real} object corresponding to an underlying reality of systems; such an object exists independently of the observer \cite{leifer2014quantum}.        

An important question is whether it is possible in principle to construct a realist $\psi$-epistemic interpretation of QM (option two), or whether a realist is compelled to think ontologically about wavefunctions. PBR try to answer this question, and their answer is negative. In their famous article \cite{pusey2012reality}, PBR introduce an operational scenario whose statistics cannot be reproduced by $\psi$-epistemic ontological models. PBR concluded that there is no way of having a realist epistemic interpretation of $\psi$ (option two); thus, the only remaining option for a realist is to choose option three.
 
The logical structure of the PBR's argumentation can be expressed in the following way. There is a contradiction between the statistics of quantum theory ($P_1$) and the statistics of $\psi$-epistemic ontological models ($P_2$) if one accepts a seemingly reasonable assumption about the ontological structure of the ontic states of a composite system, called the ``Preparation Independence Postulate" ($P_3$). In general, these statements are going to be discussed within the ontological models framework ($P_4$).  

\beq
\underbrace{ QSTAT }_{\text{$P_1$}}\hspace{2mm}\land\hspace{2mm}
\underbrace{ \psi-epistemic }_{\text{$P_2$}}\hspace{2mm}\land
\underbrace { PIP }_{\text{$P_3$}}\hspace{2mm}\land\hspace{2mm}
\underbrace{ OMF }_{\text{$P_4$}}\hspace{2mm}
\Rightarrow\hspace{2mm}\bot
\eeq

In confronting the above contradiction, so far, several arguments have been suggested aiming to put doubt on one (or even more than one) of the above propositions. Most of these suggestions can be classified into three categories.  

\begin{enumerate}
\item Rejection of $P_1$. In this category, one can argue that maybe QM's predictions are not always true. The idea of having null-measurements (measurements with no outcome) \cite{schlosshauer2014no}, or detection-inefficiency loophole \cite{detect} can be considered in this class.            
\item Rejection of $P_2$. Another option is to accept the above contradiction, as a clue of a subtle constraint on accepting the $\psi$-epistemic models. PBR themselves claimed that for a realist person, the only remaining option is to be a $\psi$-ontologist \cite{pusey2012reality}.

\item Rejection of $P_3$. Most researchers have criticized the PIP assumption. Arguing that the two subsystems with a common past may share some ontic states in the present \cite{hall2011generalisations} or the idea that there may be some holistic properties of a composite system which cannot be reduced to its sub-systems \cite{emerson2013whole}, and several other arguments which can be found in \cite{leifer2014quantum}, are examples of such criticisms.  

\item Rejection of $P_4$. Except for the arguments from antirealists, who entirely reject the existence of an underlying reality behind the quantum states, there is no serious criticism of the PBR's theorem which can be classified in this category. However, the argument provided in this study falls into this category. 
\end{enumerate}

This research highlights an incorrect conclusion in PBR's arguments, which was made due to inadequacies in the definitions of OMF. To be precise, OMF models the ontology of the preparation procedure, but it does not model the ontology of the measurement device. Such an asymmetric treatment becomes problematic, in scenarios in which measurement devices have a quantum nature. Modifying the OMF definition so that we can insert the ontology of the measurement procedure into our analysis, we can see that PBR's contradiction is not because of the $\psi$-epistemic models. Rather, the problem is in the heart of our framework (OMF). I refer this modified version of the OMF as to ``MOMF" (Modified Ontological Models Framework).  

The order of the paper's sections is as follows. A review of the OMF's definitions and how they relate to the axioms of QM, as well as a precise description of PBR's scenario in OMF, are provided in section two. Pointing out the problem of the $\psi$-epistemic ontological models caused by OMF, the mathematical structure of the MOMF formalism followed by a statement of the PBR's theorem within MOMF are discussed in section three. There, We will see how without using any further assumptions, the problem of the $\psi$-epistemic models can be solved in MOMF. In conclusion, by reviewing all the mentioned ideas, it will be argued that PBR's finding does not provide any implication for the nature of the wavefunction, rather what they have found shows us that we should be more careful about the way we model the ontology of QM. 

\section{OMF and the PBR Theorem}

\subsection{Operational QM}

Consider an operational $PM$ scenario, i.e., a scenario in which a system is prepared via a procedure/device $P$, and then measured via a procedure/device $M$. In this situation, the interaction between the prepared system and the measurement device results in an outcome for which QM provides us the associated probability of finding that outcome. Figure (\ref{PM_EXP}) depicts such a scenario. According to QM, we can express the following axioms for describing such a situation.     

\begin{figure}
\centering
\vspace{-2cm}
\includegraphics[width=0.6\textwidth]{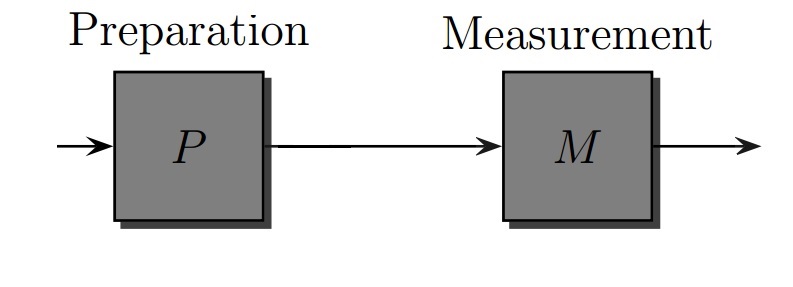}
\vspace{-5mm}
\caption{\label{PM_EXP}The schematic of an operational $PM$ scenario. A preparation device $P$, prepares a state which is supposed to be measured via a measurement device $M$.}
\end{figure}

\begin{definition}\textbf{Operational QM} posits the following two axioms to describe a $PM$ scenario \cite{ranchin2012progress}.   

\begin{axiom}[Preparation] A preparation $P$ is associated to a trace one positive operator $\rho$, known as the density operator, acting on the Hilbert space $\mathcal{H}$.
\end{axiom}

\begin{axiom}[Measurement] A measurement $M$ is associated with a positive operator valued measure (POVM) $\{ M_m \}$ such that $\sum_{m}{M_m}=I$. 

Finally, the probability of a measurement $M$ yielding outcome $m$, given a preparation procedure $P$, is

\beq
pr(m\mid P, M) = Tr(M_m \rho).
\eeq
\end{axiom}
\end{definition}

\subsection{Ontological Models Framework}

To investigate the ontological status of $\psi$, one needs an ontological framework. Such a framework cannot be inferred directly from the axioms of QM. What we can do, is to look for an ontological model whose structure is constrained by the fact that the predictions of the ontological model and the operational ones, should be equal. In 2007, Harrigan and Spekkens proposed a framework called ``Ontological Models Framework (OMF)" \cite{harrigan2010einstein} aiming to provide the most general formalism for realistic interpretations of quantum scenarios. 

The basic idea of OMF is to associate an ontic state space $\Lambda=\{\lambda_i\}$ with each physical system. This space is supposed to be the space of all the possible ontic states $\lambda_i$ of a system in the sense that the ontic states are the \textit{real} states of a system which are independent of our observation, information and what so ever. Following \cite{harrigan2010einstein}, we have the following three definitions.  

\begin{definition} \textbf{OMF} posits an ontic state space $\Lambda$ and prescribes a probability distribution over $\Lambda$ for every preparation procedure $P$, denoted by a density function $pr(\lambda\mid P)$, and a probability distribution over the different outcomes $M_m$ of a measurement $M$ for every ontic state $\lambda\in\Lambda$, denoted by a response function $pr(m\mid\lambda, M)$. Finally, for all $P$ and $M$, it must satisfy the conditions   

\beq \int_{\Lambda} pr(\lambda\mid P) d\lambda = 1 \eeq

\beq\sum_{m}{pr(m\mid\lambda, M)}=1\eeq 

\beq\int_{\Lambda} pr(\lambda\mid P) pr(m\mid\lambda,M) d\lambda = pr(m\mid P,M)=Tr(M_m\rho)\eeq
where $\rho$ is the density operator associated with $P$ and $M_m$ is the POVM element associated with the $m$-th outcome of measurement $M$.
\end{definition}

The idea of OMF is straightforward; As an operational theory, quantum theory associates an operational state $\rho$ with each preparation procedure. But, since we are not sure whether $\rho$ is the system's ontic state ($\lambda$) or not, we should simulate the ontic level of $\rho$ by the expression $pr(\lambda\mid P)$. After this simulation, the probabilities of finding different measurement outcomes will depend on the ontic states of the preparation procedure, expressed by $pr(m\mid\lambda, M)$. This implies that the outcomes' probabilities depend on the measurement procedure $M$ and the ontic state $\lambda$. Finally, these probabilities should be summed over all the relevant ontic states.                      

\begin{definition} An ontological model is \textbf{$\psi$-ontic} if for any pair of preparation procedures, $P_{\psi_1}$ and $P_{\psi_2}$, associated with distinct (and non-degenerate) quantum states $\psi_1$ and $\psi_2$, we have $pr(\lambda\mid P_{\psi_1})pr(\lambda\mid P_{\psi_2}) = 0$ for all $\lambda\in\Lambda$.
\end{definition}

\begin{definition} If an ontological model fails to be $\psi$-ontic, then it is said to be \textbf{$\psi$-epistemic}.
\end{definition}

The upper graph of figure (\ref{on-ep}) depicts the ontic state space of a $\psi$-ontic ontological model, in which, there is no shared $\lambda$ between two preparation procedures. The lower graph shows a $\psi$-epistemic ontological model. In $\psi$-ontic models, two operationally distinct states (say $\ket{x+}$ and $\ket{z+}$) are ontologically distinct too. In $\psi$-epistemic models, there exists (at least) one ontic state which is compatible with two different operational states.

\begin{figure}
\centering
\vspace{-2cm}
\includegraphics[width=0.6\textwidth]{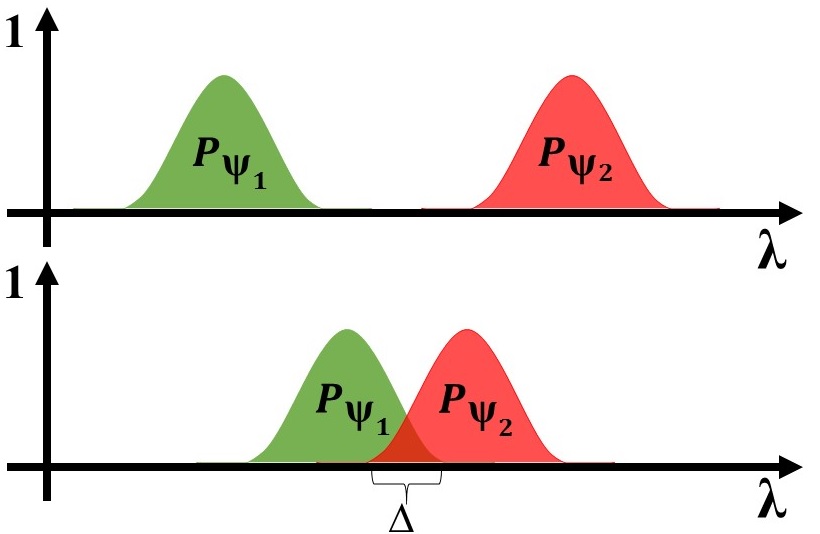}
\caption{\label{on-ep}A demonstration of $\psi$-ontic (upper graph) and $\psi$-epistemic (lower graph) ontological models. In $\psi$-epistemic models, the common support of ontic states has a non-zero measure.}  
\end{figure}

So far, all of our statements about OMF have considered the single-party $PM$ scenarios (a \textit{single} system is assumed to be prepared and then measured). In this sense, OMF aims to provide the most general ontological structure behind the operational $PM$ scenario. However, most of the problematic scenarios in foundations of QM are multipartite ones. For instance, the Bell and the PBR scenarios are both about two systems which are going to be joined to create a single composite system. 

\begin{proposition}\label{bi-omf} 
For systems $A$ and $B$ whose hypothetical composite system $AB$ is (operationally) prepared and measured by $P_{AB}$ and $M^{AB}$, \textbf{bipartite OMF} posits an ontic state space $\Lambda_{AB}$ and prescribes a probability distribution over $\Lambda_{AB}$ for every preparation procedure $P_{AB}$, denoted by a density function $pr(\lambda_{AB}\mid P_{AB})$, and a probability distribution over the different outcomes $m^{AB}$ of a measurement $M^{AB}$ for every ontic state $\lambda_{AB}\in\Lambda_{AB}$, denoted by a response function $pr(m^{AB}\mid\lambda_{AB}, M^{AB})$. Finally, for all $P_{AB}$ and $M^{AB}$, it will satisfy the following three conditions   
\beq \int_{\Lambda_{AB}} pr(\lambda_{AB}\mid P_{AB}) d\lambda_{AB}=1 \eeq
\beq\label{key-formula} 
\sum_{m^{AB}}{pr(m^{AB}\mid\lambda_{AB}, M^{AB})}=1 
\eeq 
\begin{equation}
\int_{\Lambda_{AB}} pr(\lambda_{AB}\mid P_{AB}) pr(m^{AB}\mid\lambda_{AB},M^{AB}) d\lambda_{AB} = pr(m^{AB}\mid P_{AB},M^{AB}) = Tr(M^{AB}_m\rho_{AB})
\end{equation}
where $\rho_{AB}$ is the density operator associated with $P_{AB}$ and $M^{AB}_m$ is the POVM element associated with the $m$-th outcome of measurement $M^{AB}$.
\end{proposition}

\begin{proof} This proposition is derivable from the axioms of operational QM, plus the OMF definition. Consider two systems $A$ and $B$, for which we don't know anything regarding their preparation and measurement procedures. According to the OMF definition, define an ontological model for their ``single" composite system $AB$.
\end{proof}

\begin{definition}
Consider two pure states, $\psi_A$ and $\psi_B$, prepared independently at two space-like separated points. From the OMF definition, these operational states correspond to the ontic states $\lambda_A\in\Lambda_A$ and $\lambda_B\in\Lambda_B$, respectively. If we combine these two systems to create the composite system $AB$, which is described operationally by $\psi_{AB} = \psi_A\otimes\psi_B$, the Preparation Independence Postulate \textbf{(PIP in OMF)} assumes that the ontic state of this composite system can be modeled in the following way. 
\begin{equation}\label{PIP-In-OMF-1}
\Lambda_{AB}=\Lambda_A\times\Lambda_B \equiv \lambda_{AB}=(\lambda_{A},\lambda_{B})
\end{equation}
\begin{equation}\label{PIP-In-OMF-2}
pr( \lambda_{AB} \mid \psi_{AB} ) = pr( \lambda_A \mid \psi_A) pr( \lambda_B \mid \psi_B)
\end{equation} 
where $\lambda_{AB}\in\Lambda_{AB}$ is the ontic state of the composite system $AB$.
\end{definition}

\subsection{The PBR Theorem in OMF}

Imagine Alice and Bob are living in two space-like separated regions, and they prepare their quantum states via the preparation procedures $P_A$ and $P_B$ such that there is not any operational correlation between these two processes. In the original example of PBR, Alice can choose to prepare her state in either $\ket{0_A}$ or $\ket{+_A}$. Similarly, Bob can prepare his state in either $\ket{0_B}$ or $\ket{+_B}$ (PBR express these statements, by demanding two copies of the same prepared system).

According to figure (\ref{on-ep}), an arbitrary $\psi$-epistemic ontological model implies that Alice and Bob have the overlaps $\Delta_A$ and $\Delta_B$ in the probability distributions over their ontic state spaces $\Lambda_{A}$ and $\Lambda_{B}$. Since these people are living in two space-like separated regions, there is not any operational correlation between their preparation procedures. Hence, according to the operational QM the following relation holds, 
\beq\label{p-op}
\psi_{AB} = \psi_A\otimes\psi_B.
\eeq 
So, there are four possibilities for preparing the composite system AB,
\beq
\psi_{AB}\in \{\psi^1_{AB},\psi^2_{AB},\psi^3_{AB},\psi^4_{AB}\}
\eeq where $\psi^1_{AB}=\ket{0_A}\otimes\ket{0_B}$, $\psi^2_{AB}=\ket{0_A} \otimes\ket{+_B}$, $\psi^3_{AB}=\ket{+_A}\otimes\ket{0_B}$ and $\psi^4_{AB}=\ket{+_A}\otimes\ket{+_B}$.

As it can be seen in figure (\ref{com-sup}), in general, the support of the composite ontic state $\lambda_{AB}$ is composed of nine mutually exclusive sub-regions. 
\begin{equation}\label{similarity}
supp(\lambda_{AB}\mid\psi_{AB}) = R_{11}\sqcupdot R_{12}\sqcupdot ... \sqcupdot R_{33} 
\end{equation}

Also, note that corresponding to each operational state $\psi^i_{AB}$ there exists a particular region of $\Lambda_{AB}$, to which the ontic state $\lambda_{AB}$ will be restricted. For instance, if the composite system has been prepared in $\psi^1_{AB}=\ket{0_A}\otimes\ket{0_B}$, one can infer that Alice and Bob have prepared their systems in $\ket{0_A}$ and $\ket{0_B}$, respectively. In this situation, $\lambda_{AB}$ will be restricted to the region $R_{31}\sqcupdot R_{32}\sqcupdot R_{21}\sqcupdot R_{22}$. In general, we have the following four formulas. 
\beq\label{OMF-Supports}
\begin{cases}
supp(\lambda_{AB}\mid \psi^1_{AB}) = R_{21}\sqcupdot R_{22}\sqcupdot R_{31}\sqcupdot R_{32}\\
supp(\lambda_{AB}\mid \psi^2_{AB}) = R_{11}\sqcupdot R_{12}\sqcupdot R_{21}\sqcupdot R_{22}\\
supp(\lambda_{AB}\mid \psi^3_{AB}) = R_{22}\sqcupdot R_{23}\sqcupdot R_{32}\sqcupdot R_{33}\\
supp(\lambda_{AB}\mid \psi^4_{AB}) = R_{12}\sqcupdot R_{13}\sqcupdot R_{22}\sqcupdot R_{23}\\
\end{cases}
\eeq

From the above equations, it is evident that $R_{22}$ is the only region of $\Lambda_{AB}$ which is repeating in all of the four equations. This means that $R_{22}$ is compatible with all of the four possible preparation procedures of $\psi_{AB}$.    

While the first assumption of PIP (\ref{PIP-In-OMF-1}) implies that $R_{22}=\Delta_A\times\Delta_B$, its second assumption (\ref{PIP-In-OMF-2}) implies that the ontic state $\lambda_{AB}=(\lambda_A,\lambda_B)$ of each run of the experiment will be chosen from $R_{22}$ with the non-zero probability $q=\frac {|\Delta_A\times\Delta_B|}{|supp(\lambda_{AB}\mid\psi_{AB})|}>0 $. As we will see in the next pages, $\psi$-epistemic models cannot reproduce PBR statistics because of the appearance of $R_{22}$. In other words, $R_{22}$ is a problematic region in $\psi$-epistemic models. 

So far, we have reviewed many details of PBR preparation procedure. Before explicitly expressing the main theorem, PBR measurement procedure should also be known. PBR define their measurement procedure to be $M^{AB}=\{ M^{AB}_1,M^{AB}_2,M^{AB}_3,M^{AB}_4 \}$, which is defined as follows.

\beq
\begin{aligned}
M^{AB}_1=\frac{1}{\sqrt{2}}(\Ket{0_A}\Ket{1_B}+\Ket{1_A}\Ket{0_B})\\
M^{AB}_2=\frac{1}{\sqrt{2}}(\Ket{0_A}\Ket{-_B}+\Ket{1_A}\Ket{+_B})\\
M^{AB}_3=\frac{1}{\sqrt{2}}(\Ket{+_A}\Ket{1_B}+\Ket{-_A}\Ket{0_B})\\
M^{AB}_4=\frac{1}{\sqrt{2}}(\Ket{+_A}\Ket{-_B}+\Ket{-_A}\Ket{+_B})
\end{aligned}
\eeq

It can be verified that elements of the above set constitute a complete, two-by-two orthogonal basis corresponding to four possible outcomes of a measurement device $M^{AB}$. This measurement procedure is entangled, meaning that the elements of this set cannot be separated into two parts, $M^{AB}_i \neq (M^A_i, M^B_i)$.        

\begin{theorem}[\textbf{PBR Theorem}] Assuming PIP, the (operational) PBR scenario cannot be (ontologically) reproduced by $\psi$-epistemic ontological models.
\end{theorem}

\begin{proof} To clearly show that PBR's scenario is not reproducible by any $\psi$-epistemic ontological models, let us express the arguments step-by-step.

\textbf{Step 1.} Consider a $\psi$-epistemic ontological model which is supposed to reproduce our operational statistics. Placing the PIP assumption in proposition (\ref{bi-omf}) results in the following equations for the PBR composite system.

\begin{equation}\label{OMF.7}
\begin{aligned}
p&r(m^{AB}\mid P_{AB},M^{AB}) = Tr(M^{AB}_m\rho_{AB}) =  \\
\int_{\Lambda_{AB}} &d\lambda_{AB}pr(\lambda_{AB}\mid P_{AB})pr(m^{AB}\mid\lambda_{AB},M^{AB}) = \\
\int_{\Lambda_A}\int_{\Lambda_B}&d\lambda_A d\lambda_B pr(\lambda_A\mid P_A)pr(\lambda_B\mid P_B) pr(m^{AB}\mid M^{AB},\lambda_A,\lambda_B)
\end{aligned}
\end{equation}

\begin{figure}
\centering
\vspace{-1.5cm}
\includegraphics[width=0.7\textwidth]{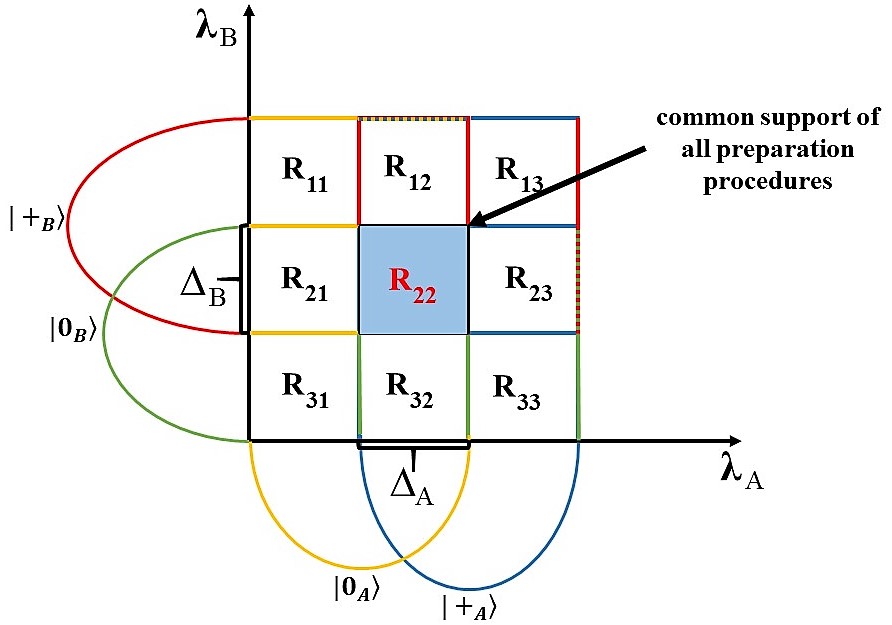}
\caption{\label{com-sup} The ontic state space $\Lambda_{AB}$ of an arbitrary $\psi$-epistemic ontological model. In general, the ontic state $\lambda_{AB}$ can be chosen from nine non-overlapped regions. The states coming from the common region $R_{22}=\Delta_A\times\Delta_B$ are compatible with all the four possible preparation procedures.}  
\end{figure}

\textbf{Step 2.} Our ontological model must be able to reproduce compatible statistics with the following two observations.

\textbf{Observation 1}: The measurement device can select only the outcomes which are not orthogonal to the prepared state. So the probability of choosing an outcome which is orthogonal to the prepared state is zero. This observation leads to four equations at the operational level of QM. 

\beq\label{equ-oper 2}\textup{operational level}
\begin{cases}
pr(m^{AB}=1\mid \psi_{AB}=\ket{0_A 0_B},M^{AB})=Tr(M^{AB}_1 \rho^{1}_{AB})=0\\
pr(m^{AB}=2\mid \psi_{AB}=\ket{0_A +_B},M^{AB})=Tr(M^{AB}_2 \rho^{2}_{AB})=0\\
pr(m^{AB}=3\mid \psi_{AB}=\ket{+_A 0_B},M^{AB})=Tr(M^{AB}_3 \rho^{3}_{AB})=0\\
pr(m^{AB}=4\mid \psi_{AB}=\ket{+_A +_B},M^{AB})=Tr(M^{AB}_4 \rho^{4}_{AB})=0\\
\end{cases}
\eeq
where $\rho^{1}_{AB}=\dyad{0_A 0_B}{0_A 0_B}$, $\rho^{2}_{AB}=\dyad{0_A +_B}{0_A +_B}$, $\rho^{3}_{AB}=\dyad{+_A 0_B}{+_A 0_B}$, and $\rho^{4}_{AB}=\dyad{+_A +_B}{+_A +_B}$. 

Combining the above results with equation (\ref{OMF.7}), one can convert this set of four operational equations into the following set of equations at the ontic level of QM. 

\beq\label{equ.ontoligcal 2}\textup{ontic level}
\begin{cases}
pr(m^{AB}=1\mid \lambda_A,\lambda_B, M^{AB})=0, \hspace{2mm} when \hspace{2mm} \psi_{AB}=\ket{0_A 0_B} \\
pr(m^{AB}=2\mid \lambda_A,\lambda_B, M^{AB})=0, \hspace{2mm} when \hspace{2mm} \psi_{AB}=\ket{0_A +_B} \\
pr(m^{AB}=3\mid \lambda_A,\lambda_B, M^{AB})=0, \hspace{2mm} when \hspace{2mm} \psi_{AB}=\ket{+_A 0_B} \\
pr(m^{AB}=4\mid \lambda_A,\lambda_B, M^{AB})=0, \hspace{2mm} when \hspace{2mm} \psi_{AB}=\ket{+_A +_B} 
\end{cases}
\eeq

Also, as explained in (\ref{OMF-Supports}), each preparation procedure of the composite system corresponds to a particular part of $\Lambda_{AB}$. Hence, domains of the last set of equations are as follows. 

\beq\label{equ.ontoligcal 3}\textup{ontic level}
\begin{cases}
pr(m^{AB}=1\mid \lambda_A,\lambda_B, M^{AB})=0, \hspace{2mm} when \hspace{2mm} \lambda_{AB}\in R_{21}\sqcupdot \textcolor{red}{R_{22}}\sqcupdot R_{31}\sqcupdot R_{32}\\
pr(m^{AB}=2\mid \lambda_A,\lambda_B, M^{AB})=0, \hspace{2mm} when \hspace{2mm} \lambda_{AB}\in R_{11}\sqcupdot R_{12}\sqcupdot R_{21}\sqcupdot \textcolor{red}{R_{22}}\\
pr(m^{AB}=3\mid \lambda_A,\lambda_B, M^{AB})=0, \hspace{2mm} when \hspace{2mm} \lambda_{AB}\in \textcolor{red}{R_{22}}\sqcupdot R_{23}\sqcupdot R_{32}\sqcupdot R_{33}\\
pr(m^{AB}=4\mid \lambda_A,\lambda_B, M^{AB})=0, \hspace{2mm} when \hspace{2mm} \lambda_{AB}\in R_{12}\sqcupdot R_{13}\sqcupdot \textcolor{red}{R_{22}}\sqcupdot R_{23}\\
\end{cases}
\eeq

\textbf{Observation 2}: According to equation (\ref{key-formula}), regardless of which quantum state has been prepared, the sum of the response functions must be equal to one.   

\beq\label{equ-oper 1}
\Sigma^{4}_{i= 1} {pr(m^{AB} = i \mid \lambda_{A}, \lambda_{B}, M^{AB})} = 1 
\eeq

\textbf{Step 3.} In this step, we check whether our model can satisfy our observations, in equations (\ref{equ.ontoligcal 3}) and (\ref{equ-oper 1}), consistently or not. First of all, remember that in $\psi$-epistemic models the ontic state $\lambda_{AB}$ comes from the the overlap region $R_{22}=\Delta_A\times\Delta_B$ at $q=\frac {|\Delta_A\times\Delta_B|}{|supp(\lambda_{AB}\mid\psi_{AB})|}$ of the times. Hence, regardless of which preparation procedure will be chosen by Alice and Bob, there is always the possibility of having the ontic states being in this region $\lambda_{AB}\in\Delta_A\times\Delta_B$. In what follows, we will see that it is impossible to simultaneously satisfy (\ref{equ.ontoligcal 3}) and (\ref{equ-oper 1}), exactly because of the appearance of $R_{22}$

Now, add the left and the right sides of the four equations of (\ref{equ.ontoligcal 3}) to each other, respectively. This leads to the following equation.
\begin{equation}\label{contradiction}
\Sigma^{4}_{i= 1} {pr(m^{AB} = i\mid\lambda_{A},\lambda_{B}, M^{AB})} = 0 \neq 1, \hspace{2mm} when \hspace{2mm} \lambda_{AB}\in \textcolor{red}{R_{22}}
\end{equation}

Comparing (\ref{equ-oper 1}) and (\ref{contradiction}), one can proof the PBR theorem; At $q$ of the times, when the ontic state $\lambda_{AB}\in R_{22}$, a $\psi$-epistemic model cannot consistently reproduce our demanded statistics.       
\end{proof}

\begin{remark}
According to the above calculations, it seems to be impossible to reproduce the statistics of the PBR scenario with $\psi$-epistemic models. PBR concluded that this contradiction reveals that $\psi$ should not be interpreted as a real epistemic object. However, in the following pages, it will be shown that this contradiction appears because there is a problem in the definition of OMF, and this contradiction has nothing to do with the nature of the wavefunction.
\end{remark}

\section{MOMF and the PBR Theorem}

A careful investigation of the given proof of the PBR theorem within OMF suggests that the problem of PBR's theorem is neither due to the nature of $\psi$ nor due to the $\psi$-epistemic models. Rather, it seems that OMF is not capable enough in describing the measurement procedures performed on the \textit{entangled} basis. According to the OMF's definition, the response functions depend on the \textit{ontic} states of the preparation device $\{\lambda\}$ and the \textit{operational} states of the measurement apparatus $\{M_m\}$. This means that, although OMF considers the ontology of the preparation procedure, asymmetrically it does not consider the ontology of the measurement device. Such a treatment becomes problematic, especially in scenarios in which measurement devices have a quantum nature. In such circumstances, the preparation, and the measurement devices create a composite system whose behavior depends on the ontic states of both $P$ and $M$. Let us denote the corresponding ontic states of $P$ and $M$ by $\lambda^P$ and $\lambda^M$, respectively. 

Needless to say, the measurement procedure of PBR's scenario is an entangled one implying that the quantum nature of $M$ is crucially important. Allowing the response functions to depend on $\lambda^P$ and $\lambda^M$, in the following subsections, we will find that there is no contradiction between the PBR operational statistics and the predictions of the $\psi$-epistemic models. 

\subsection{Modified Ontological Models Framework}

\begin{definition}\label{MOMF-First}\textbf{MOMF} posits an ontic state space $\Lambda^{PM}$ whose elements $\lambda^{PM}\in\Lambda^{PM}$ are ontologically responsible for the results of the experiment. The following three conditions must be satisfied in this framework.
\beq\int_{\Lambda^{PM}} pr(\lambda^{PM}\mid P , M) d\lambda^{PM} = 1\eeq
\beq \sum_{m} pr(m\mid\lambda^{PM}) = 1 \eeq  
\beq\label{MOMF-def}
\int_{\Lambda^{PM}} d\lambda^{PM} pr (\lambda^{PM} \mid P, M) pr (m \mid \lambda^{PM}) = pr (m \mid P, M) = Tr(M_m \rho)
\eeq 
\end{definition}

The above definition is too general and even useless in most of the cases. However, putting some reasonable constraints on this definition converts it into a more useful form. For instance, one can reasonably assume that $P$ and $M$ are operationally uncorrelated. Then, by exploiting an auxiliary assumption which mirrors the operational independence of $P$ and $M$ to the ontological level, the following statement can be obtained.

\begin{definition}\label{PM-sep} 
Within MOMF, in the absence of operational correlations between the preparation and measurement procedures, the \textbf{PM-Separability} assumption posits that the ontic state spaces of these devices, $\Lambda^P$ and $\Lambda^M$, are ontologically separable. This assumption is expressible as the conjunction of the two following sub-assumptions:  
\beq\label{A1}\Lambda^{PM}=\Lambda^P\times\Lambda^M\equiv\lambda^{PM}=(\lambda^P,\lambda^M)\eeq
\beq\label{A2} pr(\lambda^{PM}\mid P, M)=pr(\lambda^P\mid P)pr(\lambda^M\mid M)\eeq 
\end{definition}

The ideas of the PM-Separability assumption and the PIP assumption are very similar; in the absence of operational correlations between two events, one may assume that these events are ontologically uncorrelated too. As will be explained later, such auxiliary assumptions are not necessarily reasonable. However, the purpose of this paper is to show that even if we accept PBR's assumption, their conclusion does not follow.

\begin{proposition}\label{MOMF_past} Making the PM-Separability assumption has the following result. MOMF posits two ontic state spaces $\Lambda^P$, $\Lambda^M$ and prescribes two probability distributions over them for all states of $P$ and $M$, denoted density functions $pr(\lambda^P\mid P)$ and $pr(\lambda^M \mid M)$ and a probability distribution over the different outcomes $m$ of a measurement $M$ for every pair of ontic states $(\lambda^P,\lambda^M) \in\Lambda^P\times\Lambda^M$, denoted response function $pr(m\mid\lambda^P,\lambda^M)$. Finally, for all $P$ and $M$, it must satisfy the four following conditions

\beq\int_{\Lambda^P} pr(\lambda^P\mid P) d\lambda^P = 1\eeq

\beq\int_{\Lambda^M} pr(\lambda^M\mid M) d\lambda^M = 1\eeq

\beq\sum_{m} pr(m\mid\lambda^P,\lambda^M) =1\eeq  

\beq\label{two-density}
\begin{aligned}
   &pr(m \mid P,M) = Tr(M_m \rho) =\\
   \int_{\Lambda^P}\int_{\Lambda^M} d\lambda^P &d\lambda^M \underbrace{pr(\lambda^P \mid P) pr(\lambda^M \mid M)}_{\text{density functions}} \underbrace{pr(m \mid \lambda^P, \lambda^M)}_{\text{response function}}
\end{aligned}
\eeq
where $\rho$ is the density operator associated with $P$ and $M_m$ is the POVM element associated with the $m$-th outcome of measurement $M$.
\end{proposition} 

\begin{proof}
Plugging the PM-separability assumption in MOMF definition leads to this proposition. 
\end{proof}

Equation (\ref{two-density}) shows the reasonableness of the previous statements; that for simulating the ontology of an operational scenario, one needs both ontic states $\lambda^P$ and $\lambda^M$.  

Although MOMF approaches symmetrically toward the preparation and measurement procedure, there is a conceptual difference between the forms of $\Lambda^P$ and $\Lambda^M$. In principle, it is possible for different preparation methods to share some ontic states in their supports (this is the definition of $\psi$-epistemic models). However, in the measurement ontic state space $\Lambda^M$, this statement cannot hold, meaning that different regions of the ontic space (corresponding to different outcomes) are disjoint. Note that the outcomes of a measurement device correspond to different orthogonal POVM. Hence, from the operational points of view, the probability of being in an (operational) state which is compatible with two orthogonal POVMs is equal to zero. This operational claim can be confirmed by our daily experience, where we have never seen a measurement's pointer being in two different outcomes at the same time. From the ontological points of view, one should note that $\Lambda^M$ must be composed of several non-overlapping regions, corresponding to several measurement outcomes. Otherwise, we could have a set of shared ontic states $\lambda^M$,  each of which being (ontologically) compatible with (operationally) incompatible states. Hence, via mirroring the observations coming from the operational level of QM into the ontological level, we have the following assumption. Figure (\ref{measurement}) depicts this situation. 

\begin{assumption}
The ontic state space of a measurement device, $\Lambda^M$, is composed of different non-overlapping regions, $\Delta^{M}_{i}$, corresponding to its several outcomes. A measurement device with $n$ possible outcomes satisfies the following condition.
\beq supp (\lambda^M) =\Delta^M_{1}\sqcupdot\Delta^M_{2}\sqcupdot ... \sqcupdot\Delta^M_{n}\eeq
\end{assumption}

\begin{figure}
\centering
\vspace{-2cm}
\includegraphics[width=0.8\textwidth]{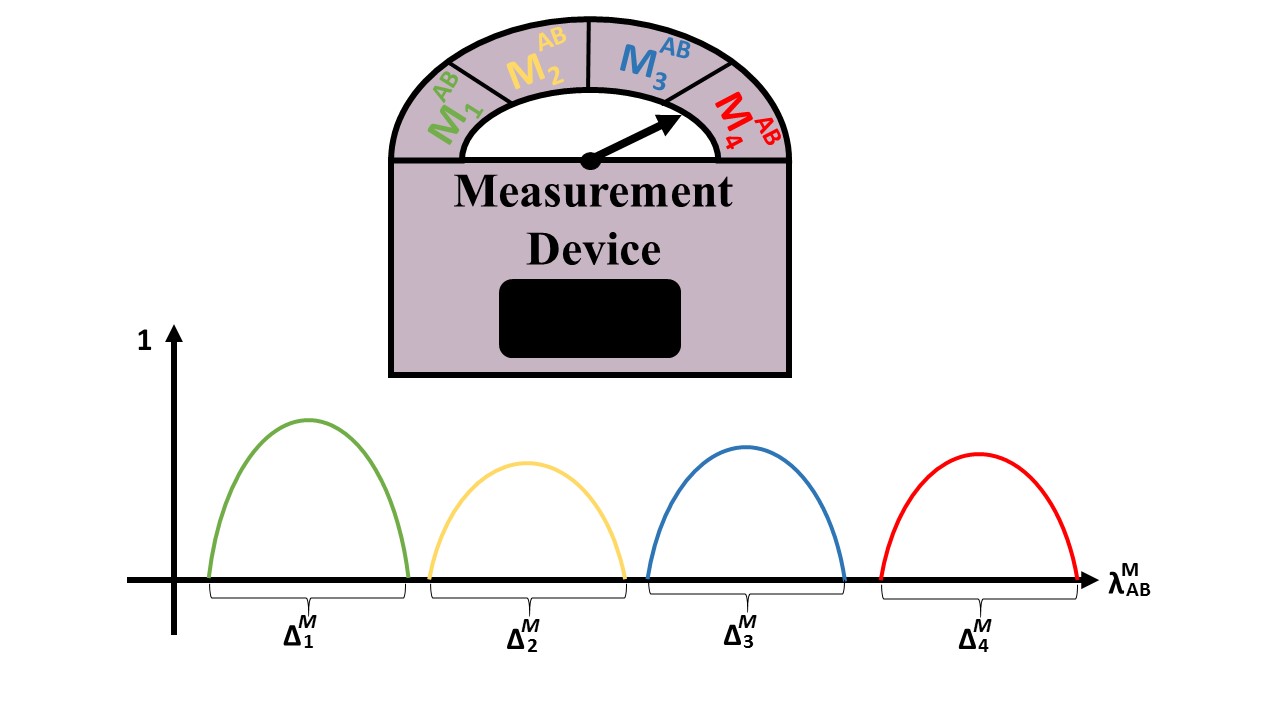}
\vspace{-5mm}
\caption{\label{measurement}The ontic space of a measurement device is composed of several non-overlapping regions corresponding to its several outcomes. In this picture, the measurement device has four outcomes whose corresponding supports, $\Delta^M_i$, are disjoint.}  
\end{figure}    

\begin{proposition}
For systems $A$ and $B$ whose hypothetical composite system $AB$ is (operationally) prepared and measured by $P_{AB}$ and $M^{AB}$, \textbf{bipartite MOMF} posits two ontic state spaces $\Lambda_{AB}^P$, $\Lambda_{AB}^M$ and prescribes two probability distributions over them for every $P_{AB}$ and $M^{AB}$, denoted density functions $pr(\lambda^P_{AB}\mid P_{AB})$ and $pr(\lambda^M_{AB} \mid M^{AB})$ and a probability distribution over the different outcomes $m^{AB}$ of a measurement $M^{AB}$ for every pair of ontic states $(\lambda^P_{AB},\lambda^M_{AB}) \in\Lambda^P_{AB}\times\Lambda^M_{AB}$, denoted response function $pr(m^{AB}\mid\lambda^P_{AB},\lambda^M_{AB})$. Finally, for all $P_{AB}$ and $M^{AB}$, it must satisfy all the following conditions

\beq\int_{\Lambda^P_{AB}} pr(\lambda^P_{AB}\mid P_{AB}) d\lambda^P_{AB} = 1\eeq

\beq\int_{\Lambda^M_{AB}} pr(\lambda^M_{AB}\mid M^{AB}) d\lambda^M_{AB} = 1\eeq

\beq\label{sharte mohem}\sum_{m^{AB}} pr(m^{AB}\mid\lambda^P_{AB},\lambda^M_{AB}) =1\eeq  

\begin{equation}\label{bi-MOMF}
\begin{aligned}
   &pr(m^{AB} \mid P_{AB}, M^{AB}) = \\
   \int_{\Lambda^P_{AB}} \int_{\Lambda^M_{AB}} d\lambda^P_{AB} d\lambda^M_{AB} pr(&\lambda^P_{AB} \mid P_{AB}) pr(\lambda^M_{AB}\mid M^{AB}) pr(m^{AB}\mid\lambda^P_{AB},\lambda^M_{AB})
\end{aligned}
\end{equation}
\end{proposition}

\begin{proof}
From \ref{MOMF-First}, define a MOMF model for the enlarged system $AB$. Make the PM-Separability assumption (\ref{PM-sep}) for $AB$.   
\end{proof}

\begin{definition}
Within MOMF, the Preparation Independence Postulate, \textbf{PIP in MOMF}, assumes the separability of the preparation ontic state space of the composite system $AB$, $\Lambda^P_{AB}$ in the following way.  
\beq\Lambda^P_{AB}=\Lambda ^P_A\times\Lambda ^P_B\equiv\lambda^P_{AB}=(\lambda ^P_A,\lambda ^P_B)\eeq
\beq pr(\lambda^P_{AB}\mid P_{AB})=pr(\lambda^P_A\mid P_A)pr(\lambda^P_B\mid P_B)\eeq
\end{definition}

\subsection{The PBR Theorem in MOMF}

Translating the details of the PBR's operational scenario in MOMF, now, we can see the reason of why the PBR's conclusion is not necessarily correct. First of all, let us explicitly express all the constraints on MOMF entailed by PBR's scenario.
\begin{enumerate}

\item We assume that the $PM$ separability assumption is valid, \beq \Lambda^{PM}=\Lambda^P\times\Lambda^M\equiv\lambda^{PM}=(\lambda^P,\lambda^M)\eeq
\beq pr(\lambda^{PM}\mid PM)=pr(\lambda^P\mid P)pr(\lambda^M\mid M)\eeq
\item According to PBR's scenario, the preparation procedures are operationally separable, \beq P_{AB}=(P_A, P_B)\eeq 
\beq \rho_{AB}=\rho_A\otimes \rho_B \eeq 
\item According to PBR's scenario, the measurement procedures are entangled, \beq M^{AB}_i\neq(M^A_i , M^B_i)\eeq.  
\item We assume that the PIP assumption is valid, 
\beq\Lambda^{P}_{AB}=\Lambda^P_A\times\Lambda^P_B\equiv\lambda^{P}_{AB}=(\lambda^P_A,\lambda^P_B)\eeq
\beq pr(\lambda^P_{AB}\mid P_{AB})=pr(\lambda^P_A\mid P_A)pr(\lambda^P_B\mid P_B)\eeq
\item Our model is $\psi$-epistemic; Some (preparation) ontic states come from a common region which is compatible with all of the four possible preparation procedures,      
\beq(\lambda^P_{A},\lambda^P_{B})\in\Delta^P_A\times\Delta^P_B\eeq 
\item The support of the measurement ontic space $\Lambda^M_{AB}$ consists of four disjoint regions corresponding to four measurement outcomes, \beq\label{pbr-measurement}\lambda^M_{AB}\in\Delta^M_{AB}=\Delta^M_1\sqcupdot\Delta^M_2\sqcupdot\Delta^M_3\sqcupdot\Delta^M_4\eeq
\end{enumerate}

\begin{remark}
Figure (\ref{2+2_dimentional.jpg}) depicts the preparation and measurement ontic state spaces $\Lambda^P_{AB}$ and $\Lambda^M_{AB}$, separately. As it was already explained, the overlap region $\Delta^P_A\times\Delta^P_B$ is a problematic region for $\psi$-epistemic OMF models. A three-dimensional representation of this problematic region is provided in figure (\ref{three-dim}). In what follows, we are going to prove that $\psi$-epistemic models can reproduce PBR statistics without any contradiction if the ontic states of the measurement device become involved.
\end{remark}

\begin{figure}
\centering
\vspace{-2cm}
\includegraphics[width=0.9\textwidth]{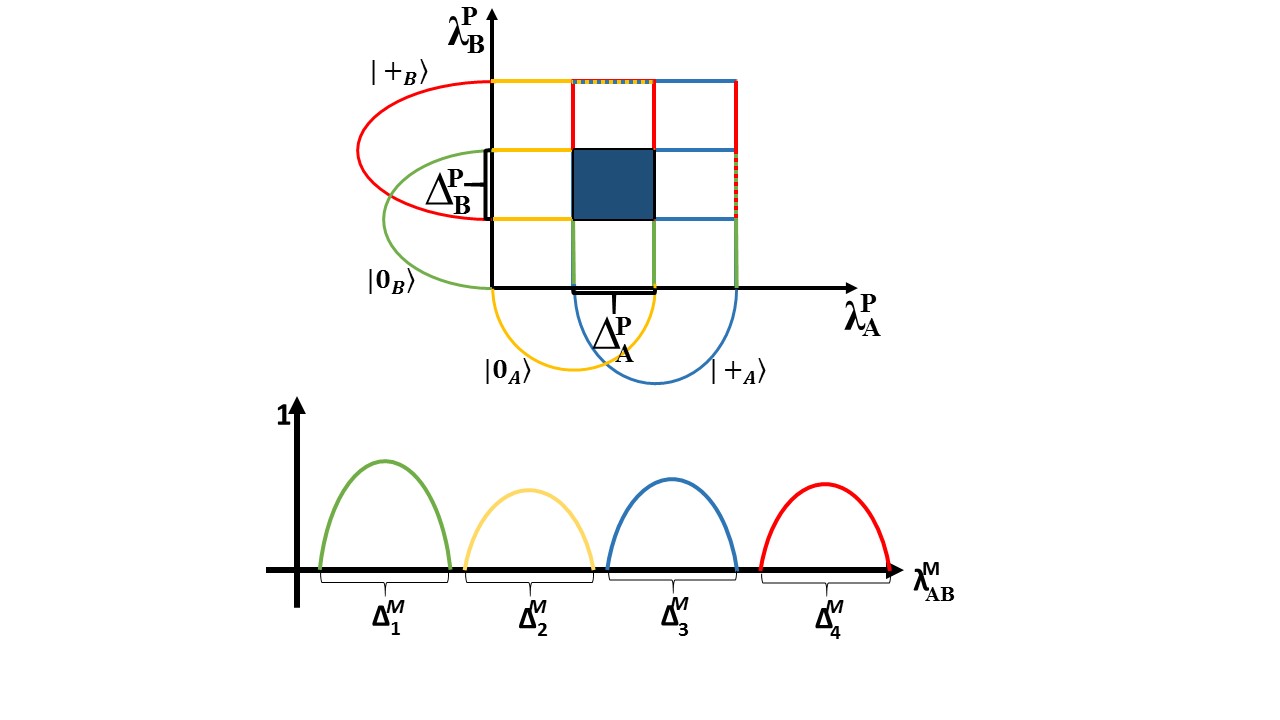}
\vspace{-1cm}
\caption{\label{2+2_dimentional.jpg}The preparation and measurement ontic state spaces $\Lambda^P_{AB}$ and $\Lambda^M_{AB}$ for the PBR's scenario. The upper graph represents the preparation ontic space, $\Lambda^P_{AB}$, which is assumed to be the Cartesian product of $\Lambda^P_{A}$ and $\Lambda^P_{B}$. The lower picture shows the measurement ontic state $\Lambda^M_{AB}$. This space is not separable into $\Lambda^M_{A}$ and $\Lambda^M_{B}$.}
\end{figure}

\textbf{Step 1.} In this step, by combing the above constraints on MOMF definition, we derive further required equations.

Making the PIP and the PM separability assumptions, the total ontic state space $\Lambda^{PM}_{AB}$ can be obtained as follows,         
\beq\label{total-space}
\Lambda^{PM}_{AB}=\Lambda^P_{AB}\times\Lambda^M_{AB}=\Lambda^P_{A}\times\Lambda^P_{B}\times\Lambda^M_{AB}\eeq

Exactly similar to (\ref{similarity}), the support of the preparation ontic state, $supp(\lambda^P_{AB})$, is composed nine mutually exclusive sub-regions. 
\begin{equation}
supp(\lambda^{P}_{AB}\mid P_{AB}) = R^P_{11}\sqcupdot R^P_{12}\sqcupdot ... \sqcupdot R^P_{33} 
\end{equation}

On the other hand, the support of the measurement ontic state, $supp(\lambda^{M}_{AB}$, is composed of four mutually exclusive sub-regions. Placing (\ref{pbr-measurement}) in (\ref{total-space}), it can be understood that the support of the total ontic state, $supp(\lambda^{PM}_{AB}$, is composed of thirty-six mutually exclusive sub-regions.
\begin{equation}\label{supports}
supp(\lambda^{PM}_{AB}\mid P_{AB} , M^{AB}) = supp(\lambda^{P}_{AB}\mid P_{AB})\times supp(\lambda^{M}_{AB}\mid M^{AB})=
\end{equation}
\begin{equation*}
(R^P_{11}\sqcupdot ...\sqcupdot R^P_{33})\times (\Delta^M_1\sqcupdot ...\sqcupdot\Delta^M_4)    
\end{equation*}

We have already seen that the problem of the $\psi$-epistemic models is the existence of the overlap region $R^P_{22}$, and not the other regions. Hence, to assess the ability of a $\psi$-epistemic MOMF model in reproducing PBR statistics, it suffices to check this ability for the states belonging to this problematic region, $\lambda^{PM}_{AB} \in R^P_{22}\times\Delta^M_{AB}$. In the total ontic state space $\Lambda^{PM}_{AB}$, there are four sub-regions being compatible with $R^P_{22} = \Delta^P_A\times\Delta^P_A$. Let us denote them by $R_1$, $R_2$, $R_3$ and $R_4$, where $R_i = R^P_{22}\times\Delta^M_{i}$. 

Figure (\ref{three-dim}) provides a three dimensional representation of these regions. Note that these regions are disjoint, and their union constructs the problematic region $R^P_{22}\times\Delta^M_{AB}$,
\beq
\Delta^P_{A}\times\Delta^P_{B}\times\Delta^M_{AB}= R_1\sqcupdot R_2\sqcupdot R_3\sqcupdot R_4
\eeq

Finally, we have the following equation for computing the probability of finding different outcomes.   
\beq\label{combine} 
pr(m^{AB} \mid P_{AB}, M^{AB}) = 
\eeq
\begin{equation*}
\resizebox{1\hsize}{!}{$\int_{\Lambda^P_A}\int_{\Lambda^P_B}\int_{\Lambda^M_{AB}}d\lambda^P_A d\lambda^P_B d\lambda^M_{AB}pr(\lambda^P_A \mid P_A) pr(\lambda^P_B \mid P_B)pr(\lambda^M_{AB} \mid M^{AB}) pr(m^{AB}\mid\lambda^P_A , \lambda^P_B , \lambda^M_{AB})$}
\end{equation*}

\textbf{Step 2.} In this step, PBR's observations will be expressed within MOMF.

\textbf{Observation 1}: Our measurement device can select the outcomes which are non orthogonal to the prepared state. So the probability of choosing an outcome which is orthogonal to the prepared state is zero. This observation leads to the four following equations on the operational level,

\beq\label{op-set}\textup{operational level}
\begin{cases}
pr(m^{AB}=1\mid \psi_{AB}=\ket{0_A 0_B},M^{AB})=Tr(M^{AB}_1 \rho^{1}_{AB})=0\\
pr(m^{AB}=2\mid \psi_{AB}=\ket{0_A +_B},M^{AB})=Tr(M^{AB}_2 \rho^{2}_{AB})=0\\
pr(m^{AB}=3\mid \psi_{AB}=\ket{+_A 0_B},M^{AB})=Tr(M^{AB}_3 \rho^{3}_{AB})=0\\
pr(m^{AB}=4\mid \psi_{AB}=\ket{+_A +_B},M^{AB})=Tr(M^{AB}_4 \rho^{4}_{AB})=0\\
\end{cases}
\eeq

Combining the above results with equations (\ref{supports}) and (\ref{combine}), one can convert this set of four operational equations into the following set of equations at the ontic (MOMF) level of QM.

\beq\label{obser1}\textup{ontic level} 
\begin{cases}
pr(m^{AB}=1\mid\lambda^P_A ,\lambda^P_B,\lambda^M_{AB})=0\hspace{2mm} when \hspace{2mm}\lambda^{PM}_{AB}\in (R^P_{21}\sqcupdot \textcolor{red}{R^P_{22}}\sqcupdot R^P_{31}\sqcupdot R^P_{32})\times (\Delta^M_1)\\
pr(m^{AB}=2\mid\lambda^P_A ,\lambda^P_B,\lambda^M_{AB})=0\hspace{2mm} when \hspace{2mm}\lambda^{PM}_{AB}\in (R^P_{11}\sqcupdot R^P_{12}\sqcupdot R^P_{21}\sqcupdot \textcolor{red}{R^P_{22}})\times (\Delta^M_2)\\
pr(m^{AB}=3\mid\lambda^P_A ,\lambda^P_B,\lambda^M_{AB})=0\hspace{2mm} when \hspace{2mm}\lambda^{PM}_{AB}\in (\textcolor{red}{R^P_{22}}\sqcupdot R^P_{23}\sqcupdot R^P_{32}\sqcupdot R^P_{33})\times (\Delta^M_3)\\
pr(m^{AB}=4\mid\lambda^P_A ,\lambda^P_B,\lambda^M_{AB})=0\hspace{2mm} when \hspace{2mm}\lambda^{PM}_{AB}\in (R^P_{12}\sqcupdot R^P_{23}\sqcupdot \textcolor{red}{R^P_{22}}\sqcupdot R^P_{23})\times (\Delta^M_4)\\
\end{cases}
\eeq

\textbf{Observation 2}: According to equation (\ref{sharte mohem}), regardless of which quantum state has been prepared, the sum of the response functions must be equal to one.
\beq\label{obser2}
\Sigma^4_{i=1} pr(m^{AB}=i\mid\lambda^P_{A},\lambda^P_{B},\lambda^M_{AB})
\eeq

\begin{figure}
\centering
\vspace{-2cm}
\includegraphics[width=.6\textwidth]{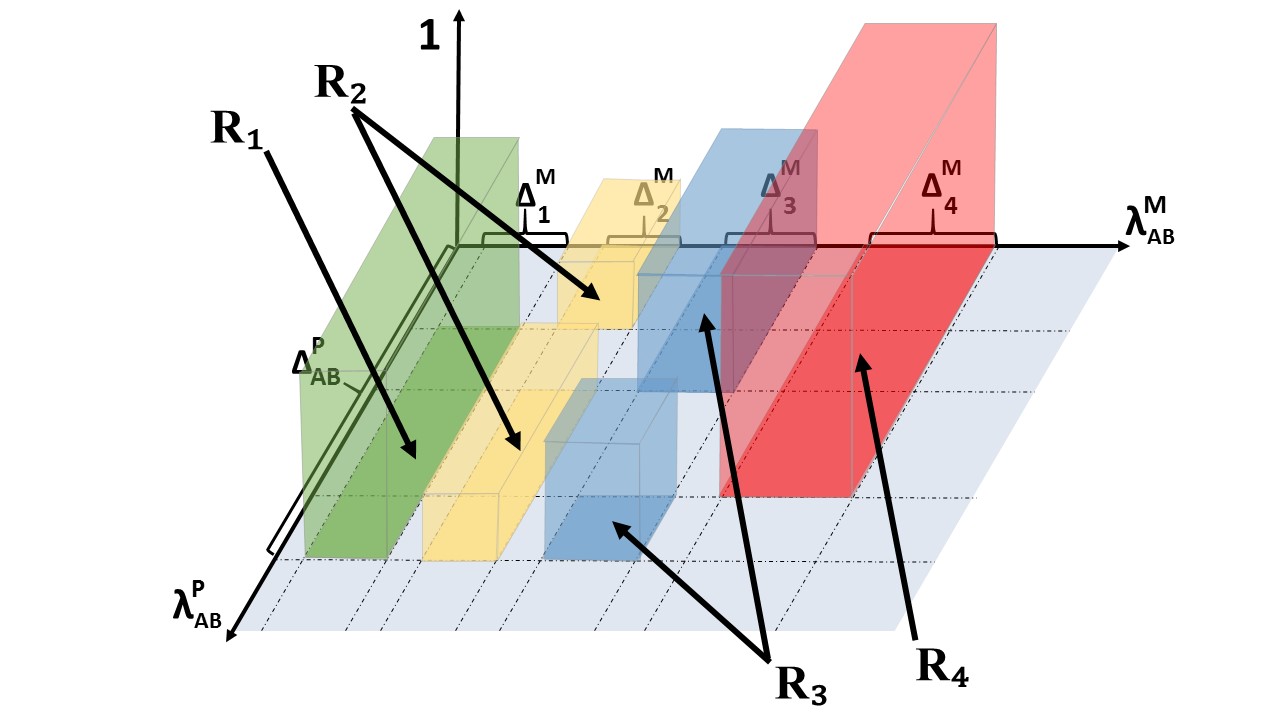}
\caption{\label{three-dim} An illustration of a specific sub-region of the total ontic state space for PBR's scenario, namely $\Delta^P_{A}\times\Delta^P_{B}\times\Lambda^M_{AB}$. This 3D-representation has been obtained from the Cartesian product of the previous 2D-representations of the preparation and measurement procedures in figures (\ref{com-sup}) and (\ref{measurement}). Clearly, this sub-region is composed of four disjoint parts, each of which being able to reproduce PBR's statistics without any contradiction.}
\end{figure}

\textbf{Step 3.} Fortunately, there is no contradiction between the implications of the first observation, in (\ref{obser1}), and the consequences of the second observation, in (\ref{obser2}):

The ontic states $\lambda^{PM}_{AB}=(\lambda^P_A,\lambda^P_B,\lambda^M_{AB})$ in equations of the set (\ref{obser1}) are not the same. Rather, each triplet belongs to one specific sub-region of the total ontic space which is disjoint from the others. Hence, if the response function of our model assigns zero to one triplet, the other three triplets become non-zero. And exactly these three triplets can reproduce the first observation (\ref{obser1}). To be precise, different parts of the measurement ontic state space are completely disjoint from each other. So, we  have
\beq |\Delta^M_i \cap\Delta^M_{j \neq i}|= 0 \eeq
Hence, supports of the ontic state in four equations (\ref{obser1}) do not share any common ontic state,
\beq
|(\textcolor{red}{R^P_{22}}\times\Delta^M_i)\cap (\textcolor{red}{R^P_{22}}\times\Delta^M_{j\neq i})|= 0  
\eeq

\section{Conclusion}
First, I investigated the PBR theorem and its underlying framework, OMF. Pointing to a problem inside the OMF definition, I defined a modified version of OMF, MOMF. Translating the PBR's scenario in MOMF, I showed that PBR's claim in rejecting the epistemic view is not because of the nature of the wavefunction. Rather, (I argued that) the reason of why PBR have found such a result lies in the OMF definition. That this, the PBR theorem has vanished within MOMF. Inserting the measurement ontic state was the key idea of this paper. This analysis suggests that the importance of PBR's work is not because of their conclusion about the nature of the wavefunction. Rather, PBR's finding reveals a significant deficiency of OMF in not considering the ontic level of measurement devices. It is important to note that, the arguments of this paper do not lead to any conclusion about the nature of wavefunction. However, the paper reveals that \textit{PBR's arguments cannot} reject the epistemic view of the wavefunction.     

\section*{Acknowledgements} Special thanks to Erik Curiel and Neil Dewar for their very careful reviews and comments. Also, I would like to thank Stephan Hartmann, Detlef D{\"u}rr, Ward Struyve and Michael Cuffaro for helpful feedback on an earlier draft. The author is supported by the MCMP. 
\bibliographystyle{unsrt}
\bibliography{Ref}
\end{document}